\documentclass[letterpaper, 10 pt, conference]{ieeeconf}  

\newcommand{\SHN}[1]{\textcolor{cyan}{#1}}
\IEEEoverridecommandlockouts                              

\usepackage{xcolor}
\usepackage{tikz}
\usetikzlibrary{patterns}
\usepackage{algorithmic}
\usepackage[ruled, lined, longend]{algorithm2e}
\usepackage{amsmath}
\usepackage{amssymb}
\newtheorem{defn}{Definition}

\newtheorem{theorem}{Theorem}

\usepackage{multirow}
\usepackage{siunitx}
\usepackage{multicol}

\DeclareMathOperator*{\argmin}{arg\ min}
\usepackage{graphicx}
\usepackage{hyperref}

\title{\LARGE \bf Learning for Online Mixed-Integer Model Predictive Control with Parametric Optimality Certificates
}

\author{Luigi Russo*, Siddharth H. Nair*,  Luigi Glielmo, Francesco Borrelli
\thanks{* denotes equal contribution. SHN and FB are with the Model Predictive Control lab, University of California, Berkeley. LR and LG are with University of Sannio. Corresponding author: \tt\scriptsize{siddharth\_nair@berkeley.edu} 
}
}

\begin{document}

\maketitle
\thispagestyle{empty}
\pagestyle{empty}

\begin{abstract}
We propose a supervised learning framework for computing solutions of multi-parametric Mixed Integer Linear Programs (MILPs) that arise in Model Predictive Control. Our approach also quantifies sub-optimality for the computed solutions. Inspired by Branch-and-Bound techniques, the key idea is to train a Neural Network/Random Forest, which for a given parameter, predicts a strategy consisting of (1) a set of Linear Programs (LPs) such that their feasible sets form a partition of the feasible set of the MILP and (2) a candidate integer solution. For control computation and sub-optimality quantification, we solve a set of LPs online in parallel. We demonstrate our approach for a motion planning example and compare against various commercial and open-source mixed-integer programming solvers.
\end{abstract}

 
\section{Introduction}\label{sec:Intro}
Multi-parametric Mixed-Integer Programming (mp-MIP) is a convenient framework for modelling various non-convex motion planning and constrained optimal control problems \cite{ioan2021mixed}. The mixed-integer formulation can model constraints such as collision avoidance \cite{marcucci2022motion}, mixed-logical specifications \cite{tokuda2021convex} and mode transitions for hybrid dynamics \cite{heemels2001equivalence}. The multi-parametric nature of these mp-MIPs arises from requiring to solve these problems for different initial conditions, obstacles configurations or system constraints---all of which affect the MIP solution.  When  Model Predictive Control (MPC)~\cite{morari1999model, borrelli2017predictive} is used for such class of problems, a MIP has to be solved  in a receding horizon fashion at each time step. However, computing solutions for MIPs is $\mathcal{NP}-$hard and challenging for real-time ($\geq$10Hz) applications.

There are two broad approaches towards solving these MIPs online for real-time MPC. The first approach is Explicit MPC \cite{borrelli2017predictive, bemporad2000optimal} which involves offline computation of the solution map of the mp-MIP explicitly as piece-wise functions over partitions of the parameter space, so that online computation is reduced to a look-up. However this approach is best suited for mp-MIPs of moderate size because the complexity of the online look-up and offline storage of partitions, increases rapidly with scale \cite{cimini2017exact}. The second approach for real-time mixed-integer MPC relies on predicting warm-starts for the mp-MIP by training  Machine Learning (ML) models on large offline datasets \cite{masti2019learning, zhu2020fast, srinivasan2021fast, bertsimas2022online, cauligi2021coco}. The authors of \cite{masti2019learning, zhu2020fast, srinivasan2021fast} use various supervised learning frameworks to predict the optimal integer variables for the mp-MIP at a given parameter so that the online computation is reduced to solving a convex program.
\begin{figure}
\vskip -0.2 true in
    \centering
    \includegraphics[width=0.8\columnwidth]{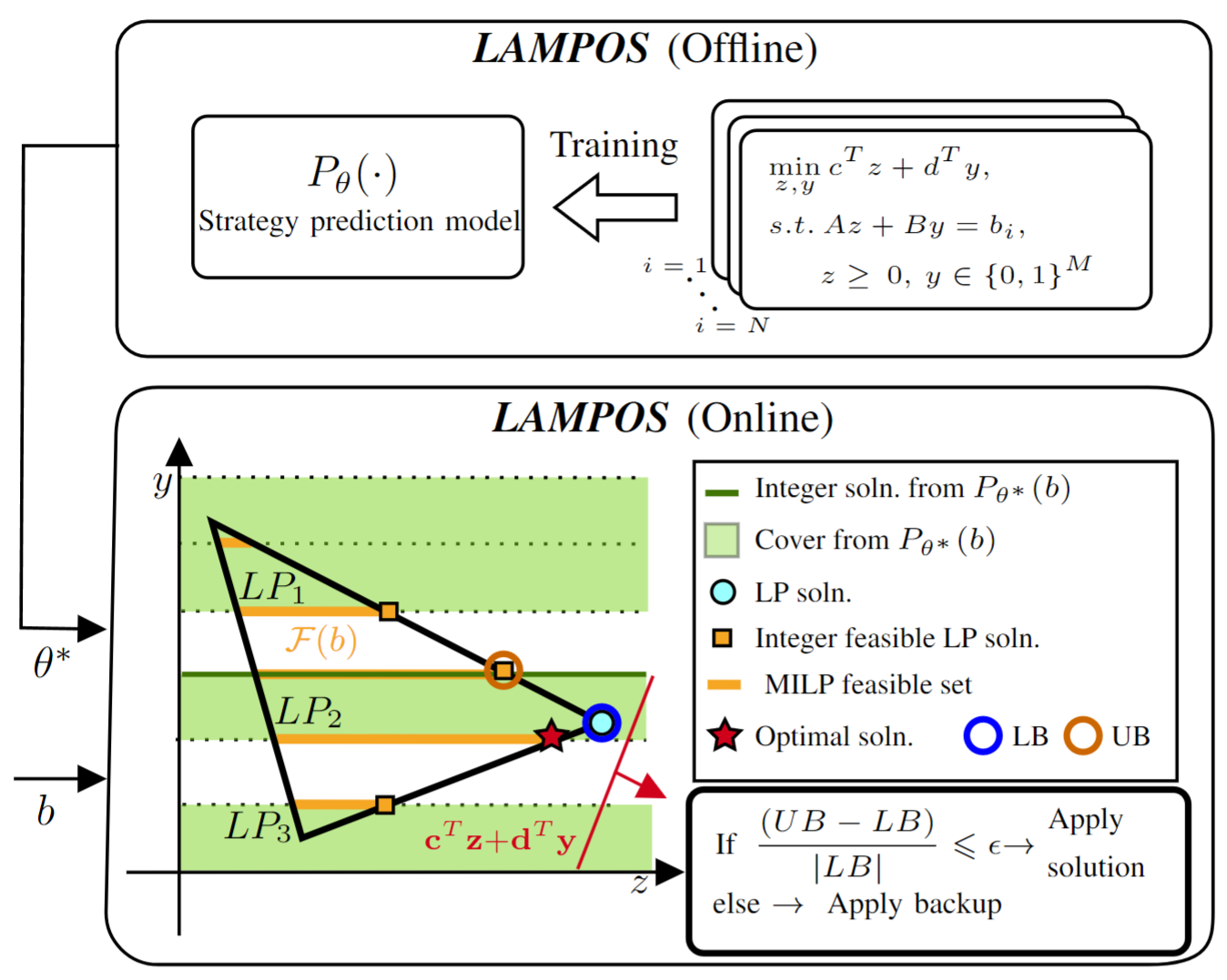}
    \vskip -0.1 true in
    {\small{
    \caption{We propose LAMPOS, a strategy-based solution approach for mp-MILPs for real-time MPC. Offline, a prediction model $P_\theta(\cdot)$ is trained on various MILP instances to learn a \text{strategy} $s(\cdot)$, mapping parameters $b$ to an optimal integer solution and a set of LPs (called a \text{cover}) obtained from the leaves of the BnB tree. Online, a solution to the MILP is obtained from the predicted strategy $s(b)$ by solving a set of LPs in parallel. The proposed strategy allows (1) sub-optimality quantification of MILP feasible solutions, and (2) recovery of MILP solution if none were found from the LPs.}}}
    \label{fig:LAMPOS_overview}
    \vskip -0.2 true in
\end{figure}
\vskip -0.1 true in
In \cite{bertsimas2022online, cauligi2021coco}, the authors define the notion of an \textit{optimal strategy} for a mp-MIP as a mapping from parameters to the complete information required to efficiently recover an optimal solution. For multi-parametric Mixed-Integer Linear/Quadratic Programs (mp-MILPs/MIQPs), an optimal strategy is defined as \textit{a set of integer variables and active constraints at the optimal solution}. Given an optimal strategy, an optimal solution can be recovered by solving a linear system of equations which is computationally inexpensive compared to tree search methods typically used for solving MIPs, such as Branch-and-Bound (BnB). Thus, a prediction model is trained offline to predict the optimal strategy for efficiently solving the MIPs online. However, a common issue that plagues these ML-based approaches is the inability to assess the quality of the predicted warm-start/strategy to guard against poor predictions, which can lead to sub-optimal or infeasible solution predictions. Indeed, prediction models may perform poorly for various reasons: insufficient richness of the model parameterization, significant shift between the training and test distribution, convergence of the training algorithm towards a sub-optimal minimum \cite{zhang2021understanding}.

In this work, we focus on mp-MILPs and propose a supervised learning framework for predicting strategies to efficiently solve the MILP, along with a mechanism to measure the sub-optimality of the prediction.
The authors of \cite{zhang2020near} propose a framework for certifying the quality of predicted solutions for parametric convex Quadratic Programs (QPs) using strong duality. The main idea is to train prediction modules offline that predict optimal solutions for both the primal QP and its Lagrangian dual. For primal and dual feasible predictions (after projection to the feasible set if necessary), the quality of the predictions can be assessed from the duality gap because of strong duality. Since strong duality does not hold in general in the framework of Lagrangian duality for MILPs, we do not adopt this approach. However we draw inspiration from \cite{zhang2020near} and the optimality certification procedure in Branch-and-Bound, to propose a strategy definition for mp-MILPs, accompanied by bounding functions to quantify the sub-optimality of the strategy. This enables us to efficiently recover the solution of a MILP online from the predicted strategy by solving some LPs online in parallel, and also measure the sub-optimality of the recovered solution. Using ideas from multi-parametric programming, we show the parametric behaviour of our proposed strategy definition. We complement this insight with a supervised learning framework for training a prediction model offline, which predicts strategies for solving the MILP online.

\section{Problem Formulation}\label{sec:pf}
Consider the general formulation for Mixed-Integer MPC (MIMPC) adapted from \cite{borrelli2017predictive}: 
\begin{align}\label{eq:MIMPC}
{\small{\begin{aligned}
V^\star(x_t)=\min_{\substack{\boldsymbol{x}_{t},\boldsymbol{u}_{t},\\\boldsymbol{\delta}_{t},\boldsymbol{z}_{t}}}\ & ||Px_{t+N|t}||_p+\sum_{k=t}^{t+N}||Q\begin{bmatrix}x_{k|t}\\\delta_{k|t}\end{bmatrix}||_p+||Ru_{k|t}||_p,\\
    \text{s.t.}\ & x_{k+1|t}=Ax_{k|t}+B_1u_{k|t}+B_2\delta_{k|t}+B_3z_{k|t},\\
    \ &E_2\delta_{k|t}+E_3z_{k|t}\leq E_1u_{k|t}+E_4x_{k|t}+E_5,\\
    \ & x_{t|t}=x_t,~\delta_{k|t}\in\{0,1\}^{n_{\delta}}~\forall k=t,..,t+N-1
    \end{aligned}}}
\end{align}
where $x_t$ is system state at time $t$,  $p=1\text{ or }\infty$, and the decision variables $\boldsymbol{x}_t=[x_{t|t},..,x_{t+N|t}]$, $\boldsymbol{u}_t$, $\boldsymbol{\delta}_t$, $\boldsymbol{z}_t$ (defined similarly) are the states, inputs, binary variables and auxiliary variables respectively. The optimal solution to \eqref{eq:MIMPC} defines the MPC policy as $\pi_{MPC}(x_t)=u^\star_{t|t}$.

The optimization problem \eqref{eq:MIMPC} can be expressed as a multi-parametric Mixed-Integer Linear Program (mp-MILP), with the parameters being the system state $x_t$. The mp-MILP can be concisely expressed as follows:
\begin{align}\label{eq:MILP}
\begin{aligned}
    V^\star(b)=\min_{z,y}\quad& c^\top z+d^\top y,\\
    \text{s.t.}\quad& Az+By=b,\\
    &z\geq 0,~y\in\{0,1\}^M
    \end{aligned}
\end{align}
with continuous decision variables $z\in\mathbb{R}^n$, binary decision variables $y\in\{0,1\}^M$ and parameters $b\in\mathbb{R}^m$. Let $z^\star(b), y^\star(b)$ be an optimal solution to \eqref{eq:MILP} and $V^\star(b)$ be the optimal cost. For a given parameter $b$, let $\mathcal{F}(b)$ be the set of $(z,y)$ feasible for \eqref{eq:MILP} and $V(b,z,y)$ be the cost of any $(z,y)\in\mathcal{F}(b)$, with sub-optimality given by $\frac{V(b,z,y)-V^\star(b)}{|V^\star(b)|}$. Also define $\mathbb{B}=\{b\in\mathbb{R}^m| \mathcal{F}(b)\neq\emptyset\}$ as the set of parameters for which \eqref{eq:MILP} is feasible.

In this work, we aim to exploit the parametric nature of the mp-MILP \eqref{eq:MILP} to predict a solution ($\tilde{z}(b),\tilde{y}(b))\in\mathcal{F}(b)$ for real-time MPC, and quantify its sub-optimality using \textit{strategies}. The strategy maps a parameter $b$ to an element of a finite and discrete set $\mathbb{S}$, which describes the complete information necessary to recover a feasible point ($z(b), y(b)$)  for \eqref{eq:MILP} (if it exists), formally defined next. 
\begin{defn}
A function $s: \mathbb{B}\rightarrow \mathbb{S}$ is a strategy for mp-MILP \eqref{eq:MILP} if there exists a map $R(\cdot)$ such that $\forall b\in\mathbb{B}: R(b, s(b))=(z(b),y(b))\in\mathcal{F}(b)$.
\end{defn}
For example, in \cite{bertsimas2022online} the set $\mathbb{S}$ is given by all possible sets of active constraints for \eqref{eq:MILP} and for each $b\in\mathbb{B}$, $s(b)$ picks the active constraints for a $(z,y)\in\mathcal{F}(b)$. The recovery map is then given as the solution of a linear system of equations.

The strategy $s^\star(b)$ is said to be \textit{optimal} at $b\in\mathbb{B}$ if $R(b, s^\star(b))=(z^\star(b), y^\star(b))$. We construct  functions $V_{\mathrm{lb}}(\cdot,\cdot)$, $V_{\mathrm{ub}}(\cdot,\cdot)$ that satisfy the following properties:
\begin{enumerate}
    \item For any  $(z,y)\in\mathcal{F}(b)$ such that $R(b,s(b))=(z,y)$,  \vskip -0.15 true in
    $$V_{\mathrm{lb}}(b,s(b))\leq V(b,z,y)\leq V_{\mathrm{ub}}(b,s(b)).$$
    \item For the optimal strategy $s^\star(b)$,
    $$V_{\mathrm{lb}}(b,s^\star(b))= V^\star(b)=V_{\mathrm{ub}}(b,s^\star(b)).$$
\end{enumerate}
 For any $b\in\mathbb{B}$, we use $V_{\mathrm{lb}}(\cdot,\cdot)$, $V_{\mathrm{ub}}(\cdot,\cdot)$ to  estimate the quality of a strategy $s(b)$ with respect to $s^\star(b)$. In particular,  the sub-optimality of a predicted strategy $\tilde{s}(b)$ is over-estimated as $\left|\frac{V_{\mathrm{ub}}(b,\tilde{s}(b))-V_{\mathrm{lb}}(b,\tilde{s}(b))}{V_{\mathrm{lb}}(b,\tilde{s}(b))}\right|$ by using the recovered solution $R(b,\tilde{s}(b))=(\tilde{z}(b),\tilde{y}(b))$.

\subsubsection*{Organization}First, we present our choice of strategy $s(\cdot)$, the recovery map $R(\cdot)$, and the bounding functions $V_{\mathrm{lb}}(\cdot)$, $V_{\mathrm{ub}}(\cdot)$ that meet the desired properties in Section~\ref{sec:strat}. Then in Section~\ref{sec:SL_MILP} we propose a supervised learning framework to approximate the optimal strategy $s^\star(b)$, and evaluate $R(b, s^\star(b))$, $V_{\mathrm{lb}}(b, s^\star(b)), V_{\mathrm{ub}}(b, s^\star(b))$ efficiently for predicting solutions to \eqref{eq:MIMPC} online, and evaluate its sub-optimality. Finally, we demonstrate our approach for motion planning using MIMPC and compare against open-source and commercial MILP solvers in Section~\ref{sec:num}.
\section{Strategy-based Solution to mp-MILPs}\label{sec:strat}
In this section, we present our design of the strategy $s(\cdot)$, the recovery map $R(b, s(b))$ and the bounding functions $V_{\mathrm{lb}}(b,s(b)), V_{\mathrm{ub}}(b,s(b))$, along with theoretical justification using ideas from the mp-MILP literature. \vskip -0.15 true in
\subsection{Preliminaries: Solving MILPs using Branch-and-Bound}
Branch-and-Bound (BnB) is a tree search algorithm that solves MILPs, with each node given as the LP sub-problem
\begin{align}\label{eq:LP_SP}
\begin{aligned}
    V^\star_{LP}(b, \mathrm{lb}, \mathrm{ub})=\min_{z,y}\quad& c^\top z+d^\top y,\\
    \text{s.t}\quad& Az+By=b,\\
    &z\geq 0,~\mathrm{lb} \leq y\leq \mathrm{ub},
    \end{aligned}
\end{align}
where the binary variable bounds $\mathrm{lb},\mathrm{ub}\in\{0,1\}^M$. For any $b\in\mathbb{R}^m$, let $\mathcal{F}_{LP}(b,\mathrm{lb},\mathrm{ub})$, $(z^\star(b,\mathrm{lb},\mathrm{ub}),y^\star(b,\mathrm{lb},\mathrm{ub}))$ denote its feasible set and optimal solution respectively. At iteration $i$ of BnB, a collection of sub-problems identified by $\mathcal{C}^i=\{\{\mathrm{lb}^i_k,\mathrm{ub}^i_k\}_{k=1}^{n_i}\}$ is maintained such that they form a \textit{cover} over the set of binary sequences $\{0,1\}^M$:
\begin{align*}\bigcup_{k=1}^{n^i}[\mathrm{lb}^i_{k},\mathrm{ub}^i_k]\supseteq\{0,1\}^M
\Rightarrow\bigcup_{k=1}^{n^i}\mathcal{F}_{LP}(b,\mathrm{lb}^i_k,\mathrm{ub}^i_k)\supseteq \mathcal{F}(b).
\end{align*}
A lower bound on $V^\star(b)$ at iteration $i$ is given as\vskip -0.2 true in
\begin{align*}
\underline{V}^i(b)=\min_{k\in\{1,\dots,n^i\}}V^\star_{LP}(b, \mathrm{lb}^i_k,\mathrm{ub}^i_k)\leq V^\star(b),
\end{align*}
which can be shown in three steps: \begin{enumerate}
\item Let $(\bar{z},\bar{y})=\argmin\{c^\top z+d^\top y| (z,y)\in\bigcup_{k=1}^{n^i}\mathcal{F}_{LP}(b,\mathrm{lb}^i_k,\mathrm{ub}^i_k)\}$ and $\bar{k}\in\{1,..,n^i\}$ be the sub-problem such that $(\bar{z},\bar{y})\in\mathcal{F}_{LP}(b,\mathrm{lb}^i_{\bar{k}},\mathrm{ub}^i_{\bar{k}})$. Then $c^\top\bar{z}+d^\top\bar{y}=V^\star_{LP}(b,\mathrm{lb}^i_{\bar{k}},\mathrm{ub}^i_{\bar{k}})$ due to global optimality of the $\bar{k}$th LP sub-problem.
\item Observe that $V^\star_{LP}(b,\mathrm{lb}^i_{\bar{k}},\mathrm{ub}^i_{\bar{k}})=\underline{V}^i(b)$, because otherwise, $\exists l\in\{1,..,n_i\}$ such that $V^\star_{LP}(b, \mathrm{lb}^i_l,\mathrm{ub}^i_l)<V^\star_{LP}(b,\mathrm{lb}^i_{\bar{k}},\mathrm{ub}^i_{\bar{k}})$, which implies the contradiction $\min\{c^\top z+d^\top y| \mathcal{F}_{LP}(b,\mathrm{lb}^i_l,\mathrm{ub}^i_l)\}<\min\{c^\top z+d^\top y|\bigcup_{k=1}^{n^i}\mathcal{F}_{LP}(b,\mathrm{lb}^i_k,\mathrm{ub}^i_k)\}$.

\item Finally since the sup-problems form a cover, $\underline{V}^i(b)=\min\{c^\top z+d^\top y |(z,y)\in\bigcup_{k=1}^{n_i}\mathcal{F}_{LP}(b,\mathrm{lb}^i_k,\mathrm{ub}^i_k)\}\leq \min\{c^\top z+d^\top y|(z,y)\in \mathcal{F}(b)\}=V^\star(b).$
\end{enumerate}

Define set of indices $\mathcal{I}^i\subseteq\{1,..,n^i\}$ such that their corresponding sup-problems have solutions that are also feasible for \eqref{eq:MILP}, i.e., 
{\small{\begin{align*}
\mathcal{I}^i=\{k\in{1,..,n^i}| (z^\star_{LP}(b,\mathrm{lb}^i_k,\mathrm{ub}^i_k),y^\star_{LP}(b,\mathrm{lb}^i_k,\mathrm{ub}^i_k))\in\mathcal{F}(b) \}.
\end{align*}}}
Then an upper bound on $V^\star(b)$ at iteration $i$ is given as,
\begin{align*}
    V^\star(b)\leq\bar{V}^i(b)=\begin{aligned}\begin{cases} \min_{k\in\mathcal{I}^i}V^\star_{LP}(b, \mathrm{lb}^i_k,\mathrm{ub}^i_k), & \mathcal{I}^i\neq\emptyset,\\
    \infty & \mathcal{I}^i=\emptyset\end{cases}\end{aligned},
\end{align*}
which is evident because $V^\star(b)\leq V^\star_{LP}(b,\mathrm{lb}^i_k,\mathrm{ub}^i_k)~\forall k\in\mathcal{I}^i$. If $\mathcal{I}^i=\emptyset$, often rounding heuristics are applied to some sub-problem solutions to produce a feasible solution in $\mathcal{F}(b)$. This describes the \textit{bounding} process of BnB.

 If $\underline{V}^i(b)\neq \bar{V}^i(b)$, then the search proceeds to the next iteration via the \textit{branching} process, which constructs a new cover $\mathcal{C}^{i+1}$ from $\mathcal{C}^i$ by splitting a sub-problem, say, $\{\mathrm{lb}^i_k,\mathrm{ub}^i_k\}$ into two new sub-problems $\{\{\mathrm{lb}^{i+1}_k,\mathrm{ub}^{i+1}_k\}, \{\mathrm{lb}^{i+1}_{k+1},\mathrm{ub}^{i+1}_{k+1}\}\}$ by fixing one or more variables to $0$ in one sub-problem, and to $1$ in the other. The branching decisions depend on $\underline{V}^i(b), \bar{V}^i(b)$, the optimal sub-problem solutions, and some tree search heuristics.

The search begins with the root node given by $\mathcal{C}^0=\{\{\mathbf{0}_M,\mathbf{1}_M\}\}$ defining the LP relaxation of \eqref{eq:MILP}. The search terminates when $\underline{V}^i(b)=\bar{V}^i(b)$ and the optimal solution is given by the feasible solution that yields  $\bar{V}^i(b)$. This optimality certificate is represented by \begin{enumerate} 
\item the optimal cover $\mathcal{C}^\star(b)=\{\{\mathrm{lb}_k^\star, \mathrm{ub}_k^\star\}_{k=1}^{n^\star}\}$ describing the LP sub-problems at the terminal iteration, 
\item the optimal binary solution $y^\star(b)$ obtained from the sub-problem corresponding to the upper-bound $\bar{V}^i(b)$.
\end{enumerate}
\subsection{Strategy Description for Parametric MILPs}
Inspired by the optimality certificate obtained from BnB, we propose the following strategy, bounding functions and recovery map:
\begin{subequations}\label{eq:strategy_bounding functions}
\begin{align}
    s(b)&=\{\mathcal{C}^\star(b), y^\star(b)\},\label{eq:sb}\\
V_{\mathrm{lb}}(b,s(b))&=\min_{k\in\{1,..,n^\star\}}V^\star_{LP}(b, \mathrm{lb}^\star_k,\mathrm{ub}^\star_k),\label{eq:LB}\\
V_{\mathrm{ub}}(b,s(b))&=\min_{\substack{Az+By^\star(b)=b,~z\geq 0}}\quad c^\top z+d^\top y^\star(b),\label{eq:UB}\\
R(b,s(b))&=\argmin_{\substack{Az+By^\star(b)=b,~z\geq 0}}\quad c^\top z+d^\top y^\star(b).\label{eq:Rb}
\end{align}
\end{subequations} The strategy $s(\bar{b})$ for parameter $\bar{b}$ is optimal if it certifies optimality of the MILP \eqref{eq:MILP} via $V_{\mathrm{lb}}(\bar{b},s(\bar{b}))=V^\star(\bar{b})=V_{\mathrm{ub}}(\bar{b},s(\bar{b}))$. The next theorem highlights the parametric behaviour of the optimality certificate provided by $s(\bar{b})$, i.e., the set of parameters $\mathcal{P}_{\bar{b}}$ for which $s(\bar{b})$ remains optimal. Thus, for any parameter $b\in\mathcal{P}_{\bar{b}}$, the optimal solution can be computed via \eqref{eq:Rb} without BnB.

\begin{theorem}\label{thm:strategy_params} Let $s^\star(\bar{b})=\{\mathcal{C}^\star(\bar{b}), y^\star(\bar{b})\}$ be the optimal strategy for solving MILP \eqref{eq:MILP} with the parameter $\bar{b}$. Then there is a set of parameters $\mathcal{P}_{\bar{b}}\subset\mathbb{B}$, given by a union of convex polyhedra for which $s^\star(\bar{b})$ is also optimal,
\begin{align*}
V_{\mathrm{lb}}(b, s^\star(\bar{b}))=V^\star(b)=V_{\mathrm{ub}}(b,s^\star(\bar{b}))~~\forall b\in\mathcal{P}_{\bar{b}} 
\end{align*}
\end{theorem}
\begin{proof}
Let $\mathcal{S}_{\bar{b}}^\star\subset\{1,.., n^\star\}$ be the set of feasible sub-problems in the cover $\mathcal{C}^\star(\bar{b})=\{\{\mathrm{lb}^\star_k, \mathrm{ub}_k^\star\}_{k=1}^{n^\star}\}$, and let $\bar{k}$ be the optimal sub-problem, for which $y^\star(\bar{b}, \mathrm{lb}^\star_{\bar k}, \mathrm{ub}^\star_{\bar k})=y^\star(\bar{b})$ and $V^\star_{LP}(\bar{b}, \mathrm{lb}^\star_{\bar k}, \mathrm{ub}^\star_{\bar k})=V^\star(\bar{b})$. 

For sub-problem $k\in\mathcal{S}_{\bar{b}}^\star$, we have from \cite[Theorems 6.2, 6.5]{borrelli2017predictive} that there exists a (convex) polyhedron of parameters $b$ given by $\mathcal{K}^k=\cup_{i=1}^{p_k}\mathcal{K}_i^{k}\subset\mathbb{B}$ such that each $\mathcal{K}_i^k$ is polyhedral, and $(z^\star(b,\mathrm{lb}_k,\mathrm{ub}_k), y^\star(b,\mathrm{lb}_k,\mathrm{ub}_k))$ are affine functions of $b$ for $b\in\mathcal{K}_i^k$. Define the set $\mathcal{Z}^k=\cup_{i=1}^{p_k}\{(z,y,b)\ |\ b\in\mathcal{K}_i^k, (z,y)=(z^\star(b,\mathrm{lb}_k,\mathrm{ub}_k), y^\star(b,\mathrm{lb}_k,\mathrm{ub}_k)) \}$ and for
 the optimal sub-problem $\bar{k}$, define the set $\mathcal{Z}^\star=\{(z,y,b)\ |\ (z,y,b)\in\mathcal{Z}^{\bar{k}},\ y=y^\star(\bar{b})\}$.

For any parameter $b\neq \bar{b}$, the solution  of sub-problem $\bar{k}$ is also optimal for the MILP \eqref{eq:LP_SP} at $b$ if 
\begin{align*}
&V^\star_{LP}(b, \mathrm{lb}^\star_{\bar{k}}, \mathrm{ub}^\star_{\bar{k}})=\min_{\substack{i\in\mathcal{S}_{\bar{b}}^\star}\backslash\{\bar{k}\}}  V^\star_{LP}(b, \mathrm{lb}^\star_{i}, \mathrm{ub}^\star_{i}),\\
&~y^\star(b, \mathrm{lb}_{\bar{k}}^\star, \mathrm{ub}_{\bar{k}}^\star))\in\{0,1\}^M
\end{align*}
and so, the strategy $s^\star(\bar{b})$ is optimal for $b$ if 
{\small{
\begin{align*}
    &V^\star_{LP}(b, \mathrm{lb}^\star_{\bar{k}}, \mathrm{ub}^\star_{\bar{k}})=V_{\mathrm{lb}}(b, s^\star(\bar{b})),~ y^\star(b, \mathrm{lb}_{\bar{k}}^\star, \mathrm{ub}_{\bar{k}}^\star))=y^\star(\bar{b})\\
    \Leftrightarrow &~c^\top z^{\bar{k}}+d^\top y^{\bar{k}} \leq c^\top z^{k}+d^\top y^{k},\\&~~(z^{\bar{k}},y^{\bar{k}},b)\in\mathcal{Z}^\star,(z^{k},y^{k},b)\in\mathcal{Z}^k~\forall k\in\mathcal{S}_{\bar{b}}^\star\backslash\{\bar{k}\}.
\end{align*}}}
Thus, the set of parameters for which $s^\star(\bar{b})$ is the optimal strategy is given by the set
{\small{
\begin{align*}
\mathcal{P}_{\bar{b}}=\left\{b\ \middle\vert\begin{aligned}\exists &(z^{\bar{k}},y^{\bar{k}},b)\in\mathcal{Z}^\star,\\\exists &(z^{k},y^{k},b)\in\mathcal{Z}^k~\forall k\in\mathcal{S}_{\bar{b}}^\star\backslash\{\bar{k}\}:\\&c^\top z^{\bar{k}}+d^\top y^{\bar{k}} \leq c^\top z^{k}+d^\top y^{k}
\end{aligned}\right\}
\end{align*}}}
which is a union of convex polyhedra ($\because$ affine projection of unions of convex polyhedra $\mathcal{Z}^\star, \mathcal{Z}^k$,  intersected by  affine halfspaces $c^\top z^{\bar{k}}+d^\top y^{\bar{k}} \leq c^\top z^{k}+d^\top y^{k}$).
\end{proof}
The sets $\mathcal{P}_{\bar{b}_i}$ can be constructed using ideas from multi-parametric programming, but this approach would become intractable as the size of the problem increases. Instead, we propose a supervised classification approach to predict an optimal strategy for a given parameter in the next section. For a predicted strategy  $\tilde{s}(b)$, the functions \eqref{eq:LB},\eqref{eq:UB} are used to quantify its sub-optimality compared to $s^\star(b)$. If no feasible solution is found or the predicted strategy is too sub-optimal, an optimal solution can be retrieved from $\tilde{\mathcal{C}}(b)$ by solving MILP sub-problems.
\section{LAMPOS: Learning-based Approximate MIMPC with Parametric Optimality Strategies}\label{sec:SL_MILP}
 This section presents LAMPOS: (A) an offline supervised learning framework for strategy prediction, and (B) an online algorithm for finding solutions to \eqref{eq:MIMPC}. The learning problem of predicting $s^\star(b)$ is split into two classification problems, from parameters $b$ to corresponding labels $(\gamma^\star,\upsilon^\star)$ for optimal cover $\mathcal{C}^\star(b)$ and binary solution $y^\star(b)$, respectively. The number of possible strategies/labels is exponential in the problem size, which would make the classification problem intractable as well. To address this issue, we construct our dataset with a limited number of strategies using the approach in \cite{bertsimas2022online}. For online deployment, the predicted strategy is used to obtain solutions to the \eqref{eq:MIMPC} using $R(\cdot)$, with sub-optimality quantification using $V_{\mathrm{lb}}(\cdot)$, $V_{\mathrm{ub}}(\cdot)$.
 \subsection{Offline Supervised Learning for Strategy Prediction}\label{ssec:LAMPOS_offline}
\subsubsection{Dataset Construction}\label{ssec:dataset_cnsrtct}
Our dataset consists of parameter-strategy pairs $(b_i,s(b_i))$ where the strategy $s(b_i)=(\gamma_i,~\upsilon_i)$ consists of a tuple of labels.
To determine the required number of strategies, given $M$ strategies $\mathcal{S}(\mathcal{B}_K)=\{s_1,s_2,...,s_M\}$ corresponding to $N$ independent parameter samples $\mathcal{B}_N=\{b_1,b_2,...b_N\}$, we assess the probability of encountering a new strategy with a new i.i.d. sample $b_{N+1}$, i.e., $ 
 \mathbb{P}(s(b_{N+1}) \notin \mathcal{S}(\mathcal{B}_N))$. As in \cite{bertsimas2022online}, we adopt the Good-Turing estimator $G = N_1/N$, where $N_1$ represents the number of strategies that have appeared exactly once, to bound this probability with confidence at least $1-\beta$ as:
{\small{\begin{align*}
\mathbb{P}\left(s\left(b_{N+1}\right) \notin \mathcal{S}\left(\mathcal{B}_N\right)\right) \leq G +c \sqrt{\frac{1}{N} \ln \left(\frac{3}{\beta}\right)}
\end{align*}}\vskip -0.1 true in}
where $c=(2\sqrt{2}+\sqrt{3})$. For a fixed confidence $\beta<<1$, we sample strategies and update $G$ until the right-hand side bound is less than a desired probability guarantee $\epsilon>0$.


\subsubsection{Architecture and Learning problem}

The classification problem for predicting the strategy, can be solved using popular prediction architectures such as Deep Feedforward Neural Networks (DNN) and Random Forests (RF), discussed as follows.
\subsubsection*{DNN-based Architecture}The DNN for cover prediction comprises $L$ layers composed together to define a function of the form
$\hat{\gamma}= f_L(h_{L-1}(...f_1(b)))$.  
The output of the $l$th layer is given by $y_l = f_l(y_{l-1}) = \sigma_l(W_l y_{l-1} + b_l)$
where $W_l \in \mathbb{R}^{n_l \times n_{l-1}}$ and $b_l \in \mathbb{R}^{n_l}$ are the layer's parameters, $y_0=b$, $y_L=\hat{\gamma}$ and $\sigma_l$ is the activation function used to model nonlinearities. For binary solution prediction for the MIMPC, we express $y^*(b)=[y^*_1(b),y^*_2(b),...,y^*_N(b)]$ to divide the classification problem into $N$ sub-problems,  corresponding to each step along the horizon $N$. Each sub-problem $j\in(1,2.., N)$ consists of finding the label $\nu_{j}  \in \Upsilon_j$ associated with the binary solution for step $j$ correspondent to the input parameter $b$, where $\Upsilon_{j}$ is the set of labels for sub-problem $j$, and is solved using a $K$ layer DNN with parameters $W_k \in \mathbb{R}^{n_k \times n_{k-1}}$ and $b_k \in \mathbb{R}^{n_k}$ returning a label estimation $\hat{\nu_{j}}$. The label for $\hat{y}(b)$ is given by the vector of labels $\hat{\upsilon}=[\hat{\nu_{1}},\hat{\nu_{2}},...,\hat{\nu_{N}}]$. This architecture makes the classification task easier than directly recovering the full binary solution $y^*(b)$  due to the high number of different binary solutions in the dataset. 
The training process for DNN consists of finding the network parameters that minimize a loss function that encodes misclassification error. For all the classification problems, the Cross Entropy loss function is chosen, defined as $H(p,q) = -\sum_i p_i \log(q_i)$, where $p$ is the true label distribution and $q$ is the predicted label distribution. The optimization problem for DNN training is solved using Stochastic Gradient Descent (SGD).
\vskip -0.15 true in
\subsubsection*{RF-based Architecture} \label{sss:RF-arch}The RF consists of multiple decision trees that are trained on random subsets of the training data, and the final prediction is made by aggregating the predictions of the individual trees. For the classification problems for binary and cover prediction, the Gini impurity criterion can be used as the splitting criterion, which measures the degree of impurity in a set of labels. The Gini impurity is defined as $\text{Gini}(p) = \sum_{i=1}^{K}p_i(1-p_i)$ where $p_i$ is the fraction of samples in a given set that belong to class $i$. The Gini impurity is minimized by selecting the split of the parameter space that maximizes the reduction in impurity, which is known as the greedy approach.%




\subsection{Online Deployment for MIMPC}
After training the prediction models offline, the online deployment of our approach for MIMPC is described in Algorithm ~\ref{algo:online_LMIMPC}. The inputs to the algorithm are the trained strategy prediction model $P_{\theta^\star}(\cdot)$, the state of system $x_t$ and the desired sub-optimality tolerance $tol$. 
{\small{
\begin{algorithm}[h]
    \caption{LAMPOS (Online)}\label{algo:online_LMIMPC}
    \DontPrintSemicolon
    \SetKwInOut{KwIn}{Input}
    \SetKwInOut{KwOut}{Output}
    \SetKwFunction{backup}{find\_sol}
    \SetKwFunction{MIMPC}{solve\_MIMPC}
    \KwIn{$P_{\theta^\star}(\cdot)$, $x_t$, $tol$}
    \KwOut{$\pi_{MPC}(x_t)$}   
    \SetKwProg{Fn}{Procedure}{$(x_t)$:}{\KwRet $\pi_{MPC}(x_t)$}
    \SetKwProg{Bu}{Backup}{$(\{(V_k,\mathrm{lb}_k,\mathrm{ub}_k)\}_{k=1}^{n_c})$:}{\KwRet $(\bar{V},\bar{z},\bar{y})$}
    \Fn{\MIMPC}{
    \tcc{Predict strategy}
    $[\ \tilde{\mathcal{C}}(x_t):=\{\{\tilde{\mathrm{lb}}_k,\tilde{\mathrm{ub}}_k\}\}_{k=1}^{n_c},~~ \tilde{y}(x_t)  ]\leftarrow P_{\theta^\star}(x_t)$\\
    \tcc{Add LP for fixing $\tilde{y}(x_t)$}
    $\tilde{\mathcal{C}}(x_t)\leftarrow \tilde{\mathcal{C}}(x_t)\cup \{(\tilde{y}(x_t),\tilde{y}(x_t))\}$\\
    \tcc{Solve LPs in parallel} 
    $\tilde{\mathcal{I}}\leftarrow \emptyset$, $\pi_{MPC}(x_t)\leftarrow \emptyset$ \\
    \textbf{par}\For{$k=1$ \KwTo $n_c+1$}{
         $(V_k, z_k, y_k)\leftarrow$\text{solve\_LP}$(x_t, \tilde{\mathrm{lb}}_k, \tilde{\mathrm{ub}}_k)$ \\
        \tcc{Collect MILP feasible $k$s}
         \uIf{$y_k\in\{0,1\}^M$}{
         $\tilde{\mathcal{I}}\leftarrow \tilde{\mathcal{I}}\cup\{k\}$
         }
        }
        \tcc{Check sub-optimality}
        $\{(\bar{V}_k,\bar{\mathrm{lb}}_k,\bar{\mathrm{ub}}_k)\}_{k=1}^{n_c}\leftarrow \text{sort}(\{(V_k,\tilde{\mathrm{lb}}_k,\tilde{\mathrm{ub}}_k)\}_{k=1}^{n_c})$\\
        $LB=\bar{V}_1$\\
        \uIf{$\tilde{\mathcal{I}}\neq \emptyset$}{$UB=\min_{k\in\tilde{\mathcal{I}}}V_k$, $z^\star\leftarrow z_{\argmin_{k\in\tilde{\mathcal{I}}}V_k}$\\
        \uIf{$\text{UB}-\text{LB}\leq tol\cdot|\text{LB}|$}{$\pi_{MPC}(x_t)=Sz^\star$}}
        \uIf{$\pi_{MPC}=\emptyset$}{
        \tcc{Call backup} 
        $(\bar{V}, \bar{z},\bar{y})\leftarrow$\backup($\{(\bar{V}_k,\bar{\mathrm{lb}}_k,\bar{\mathrm{ub}}_k)\}_{k=1}^{n_c}$)\\
        $\pi_{MPC}(x_t)=S\bar{z}$}
        }
        \Bu{\backup}{
        $V_{n_c+1}\leftarrow\infty, 
 ~(\bar{V},\bar{z},\bar{y})\leftarrow(\infty,\emptyset,\emptyset)$ \\
        \For{$k=1$ \KwTo $n_c$}{
        $(\hat{V}_k, \hat{z}_k, \hat{y}_k)\leftarrow$\text{solve\_MILP}$(x_t, \mathrm{lb}_k, \mathrm{ub}_k)$ \\
        $(\bar{V},\bar{z},\bar{y})\leftarrow\text{best\_sol}(\{(\hat{V}_i, \hat{z}_i,\hat{y}_i)\}_{i=1}^k)$\\
        \uIf{$\bar{V}\leq V_{k+1}$}{
        \textbf{break}}}
        }
\end{algorithm}
}}
The function \texttt{solve\_MIMPC}$(\cdot)$ returns the MPC policy $\pi_{MPC}(\cdot)$. Inside it, we first query the prediction model at the current state to obtain a strategy consisting of the cover $\tilde{\mathcal{C}}(x_t)$ and a candidate binary solution $\tilde{y}(x_t)$. The list of LP sub-problems in the cover is augmented with another LP by fixing the binary variable bounds to $\tilde{y}(x_t)$. Then the LPs are solved in parallel, while keeping track of MILP feasible solutions. The solved sub-problems are sorted in the increasing order of cost, with $\infty$ assigned to the cost of infeasible LPs. The lower bound $LB$ on the optimal cost is provided by the first LP sub-problem. The upper bound $UB$ is obtained from the best MILP feasible solution, if any. If the estimated sub-optimality $\frac{UB-LB}{|LB|}$ is within tolerance, the MPC policy is obtained as $Sz^\star$ where $z^\star$ is the LP solution corresponding to the upper bound and $S$ is a matrix that selects $u^\star_{t|t}$ from $z^\star$. If no MILP feasible solutions were found (meaning $\tilde{\mathcal{I}}=\emptyset$) or the predictions don't meet the sub-optimality tolerance, we send the sorted LP sub-problems to the backup procedure \texttt{find\_sol}$(\cdot)$ which solves a sequence of MILP sub-problems. The backup returns an optimal solution if the MILP \eqref{eq:MIMPC} is feasible, and nothing otherwise.
 \vskip -0.2 true in
\section{Numerical Experiments}\label{sec:num}
In this section we demonstrate the effectiveness of our approach for a motion planning problem and compare the performance against MILP solvers: GLPK-MI \cite{makhorin2008glpk}, SCIP \cite{achterberg2009scip}, Mosek \cite{aps2019mosek} and Gurobi \cite{gurobi}. Our implementation is available at: {\small{\SHN{ \url{https://github.com/shn66/LAMPOS}}}}.
\subsection{MIMPC for 2D Motion planning}  \label{ss:motion_pl}
\begin{figure}[!h]
 \label{fig:2Dmp}
 \vskip -0.2 true in
    \centering  
   \includegraphics[width=0.7\columnwidth]{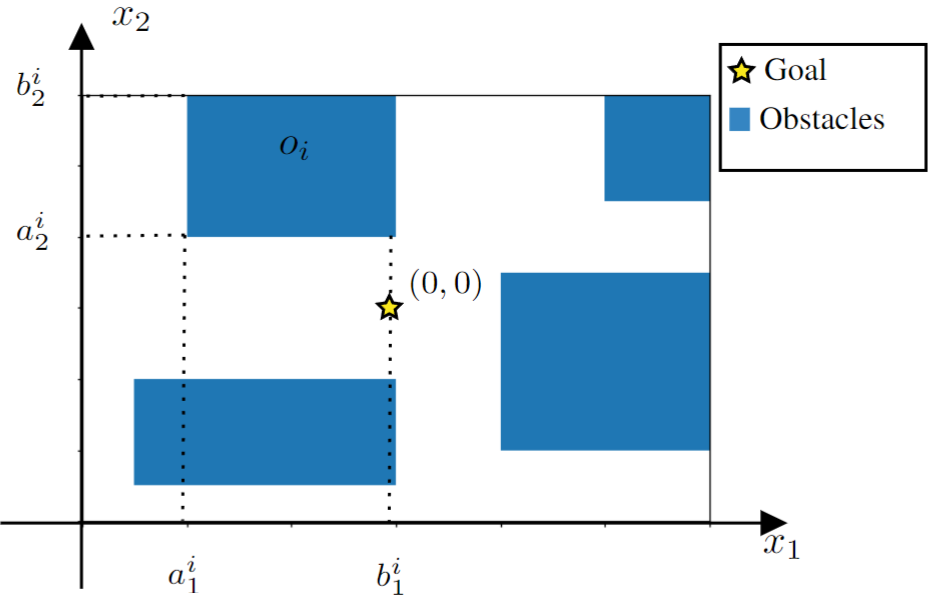}
    \vskip -0.2 true in
    {\small{
    \caption{Obstacle configuration of the 2D motion planning problem, with the $i$th  obstacle's shape:  $\{(X,Y)|~ [a^i_1,a^i_2]\leq [X,Y]\leq [b^i_1,b^i_2]\}$.}}}
    \vskip -0.25 true in
\end{figure}
The motion planning problem is to steer the robot to the origin subject to state-input and obstacle avoidance constraints as depicted in Fig.~1. The robot is modelled as a double integrator (Euler discretized at dt$=0.1s$) with state $x_t$, position $Cx_t$. Policy $\pi_{MPC}(x_t)$ is computed by solving the mp-MILP \eqref{eq:2D_motionplanning}, which is parametric in $x_t$. The obstacle avoidance constraints are encoded using the big-M method, with binary vectors $\underline{\delta}^i_{k|t},\bar{\delta}^i_{k|t} \in \{0,1\}^2$ introduced for each obstacle $i$ at time $k$, totalling  $4\cdot n_{obs}\cdot N$ binary variables for a prediction horizon of $N$ and $n_{obs}=4$ obstacles. The vectors $\bar{X}=-\underline{X}=[3,3,2,2], \bar{U}=-\underline{U}=[2,2]$ define the state-input constraints in \eqref{eq:2D_motionplanning}, and $Q=10^3 I_4,~ R=50 I_2, ~P=10^5I_4$ define the costs matrices. We model the problem using CVXPY and perform experiments for $N=20,40$. 
\begin{align}\label{eq:2D_motionplanning}
{
\small{\begin{aligned}
\min_{\mathbf{x}_t, \mathbf{u}_t, \boldsymbol{\delta}_t} & ~\left\|P x_{t+N|t}\right\|_\infty+\sum_{k=t}^{t+N-1} \left\|Qx_{k|t}\right\|_\infty+\left\|Ru_{k|t}\right\|_\infty \\ 
\text {s.t. } & ~~x_{k+1|t}=Ax_{k|t}+B u_{k|t},\\ 
& ~~\underline{X}\leq x_{k+1|t}\leq \bar{X},~ \underline{U}\leq u_{k|t}\leq\bar{U}, \\
& ~~b^i-\bar{\delta}_k^i M \leq Cx_{k+1|t} \leq a^i+M \underline{\delta}_k^i,\\
& ~~\mathbf{1}^\top\underline{\delta}^i_k+\mathbf{1}^\top\bar{\delta}^i_k\leq 3, \\
& ~~\underline{\delta}^i_k,\bar{\delta}^i_k\in\{0,1\}^2 ~ \forall i=1,..,n_{obs}\\
&~~ x_{t|t}=x_t,~~~~~~~\forall k=t,..,t+N-1 
\end{aligned}
}}
\end{align}
\subsection{Implementation Details}
\subsubsection{Dataset construction}For dataset construction we used SCIP optimization toolkit that allows us to access and save node information during the construction of the BnB tree for the MILP solution\footnote{If the LP sub-problems at the leaves of the BnB tree are unavailable, we provide a recursive algorithm  (@\SHN{github repo}) to construct a cover $\{\{\mathrm{lb}_k, \mathrm{ub}_k\}\}_{k=1}^{n_c}$ given the optimal sub-problem $\{\mathrm{lb}^\star, \mathrm{ub}^\star\}$, and a partial list of sub-problems $\{\{\mathrm{lb}_k, \mathrm{ub}_k\}\}_{k=1}^{n_p}$. The algorithm proceeds by adding disjoint \textit{facets} $[\mathrm{lb}_i, \mathrm{ub}_i] \subset [0,1]^M$ until $\cup_{k=1}^{n_c}[\mathrm{lb}_k, \mathrm{ub}_k]\supset \{0,1\}^M$.}. We randomly sample parameters $b=x_t$  and solve \eqref{eq:2D_motionplanning}. For each $b_i$ we collect the set of leaves of the BnB tree that represent our optimal cover $\mathcal{C}^*(b_i)$ and the binary solution $y^*(b_i)$. For eliminating redundant strategies we search within the dataset locally around $b_i$ and look for strategies $s(b_j)$ for which the optimality is maintained for $b_i$. For meeting the probability bound defined in Sec. \ref{ssec:dataset_cnsrtct}, we fix $\beta=10^{-3}, \epsilon=10^{-1}$. After data collection, we further process the dataset by reassigning strategies with covers with large number of LP sub-problems, to another strategy with the least sub-optimality and with fewer LP sub-problems than a pre-defined threshold (to limit the online computation).
 \subsubsection{Supervised learning} For strategy predictions, we use RF for the $N=20$ case and DNN for the $N=40$ case. For RF implementation we used the \textit{RandomForestClassifier} from \textit{sci-kit} setting  number of trees $n_t=10$ and used weighted tree splitting for both cover and binary solution classification to mitigate unbalanced-ness in the dataset. The RFs were trained until prediction accuracies $\geq 97\%$ are achieved for binary and cover predictions. We use Pytorch for our DNN implementation with architectures given by 2 hidden layers with width 64 for binary prediction, and 3 hidden layers with width 128 for cover prediction.
\subsection{Results}
We tested our approach for cases $N=20, 40$ by sampling different initial conditions $x_0$ and solve \eqref{eq:2D_motionplanning} to get the policy $\pi_{MPC}(\cdot)$ until the robot reaches the origin. For Algorithm \ref{algo:online_LMIMPC}, the LPs were solved using ECOS \cite{domahidi2013ecos} (for its rapid infeasibility detection) and SCIP for the backup MILP sub-problems. We compare LAMPOS against GLPK\_MI, SCIP, Mosek and Gurobi for solve-times. The solve-times of our approach are compared to other solvers in Fig. \ref{fig:solve_time_cmp},~\ref{fig:solve_time_cmpN40}.
 Our solve-times include prediction time and LP sub-problem solution times. In addition, we also solve \eqref{eq:2D_motionplanning} with SCIP, Mosek and Gurobi with a time-limit of 50ms for $N=20,40$ \footnote{No such interface for GLPK\_MI in CVXPY}, and compare against our approach for sub-optimality of the obtained solution (if any).  For each solver, we report the average sub-optimality of feasible solutions found and  \% of instances that it timed-out. 
 \begin{figure}[!h]
\vskip -0.1 true in
    \centering
    \includegraphics[width=0.7\columnwidth]{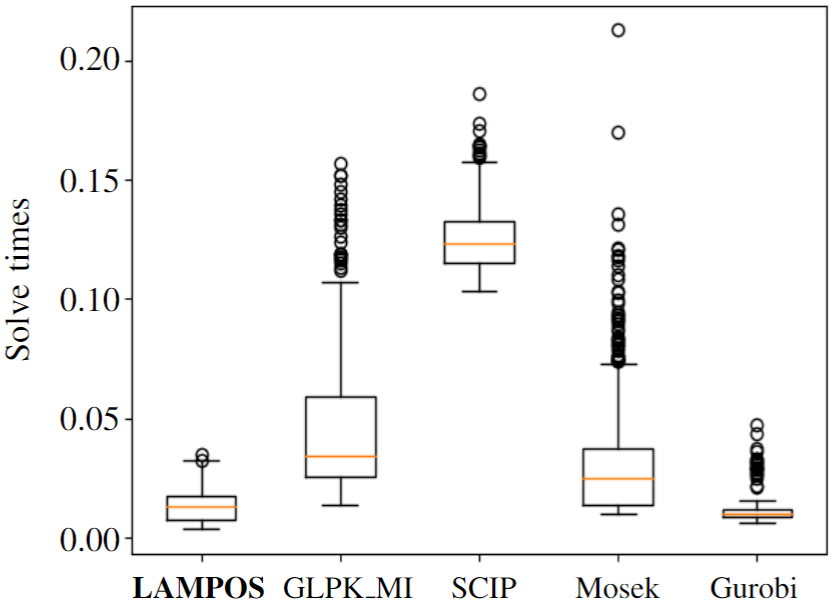}
    \vskip -0.15 true in
    \caption{Comparison with solution times of other solvers  for $N=20$}
    \label{fig:solve_time_cmp}
    \vskip -0.15 true in
\end{figure}
\begin{figure}[!h]
\vskip -0.1 true in
    \centering
    \includegraphics[width=0.68\columnwidth]{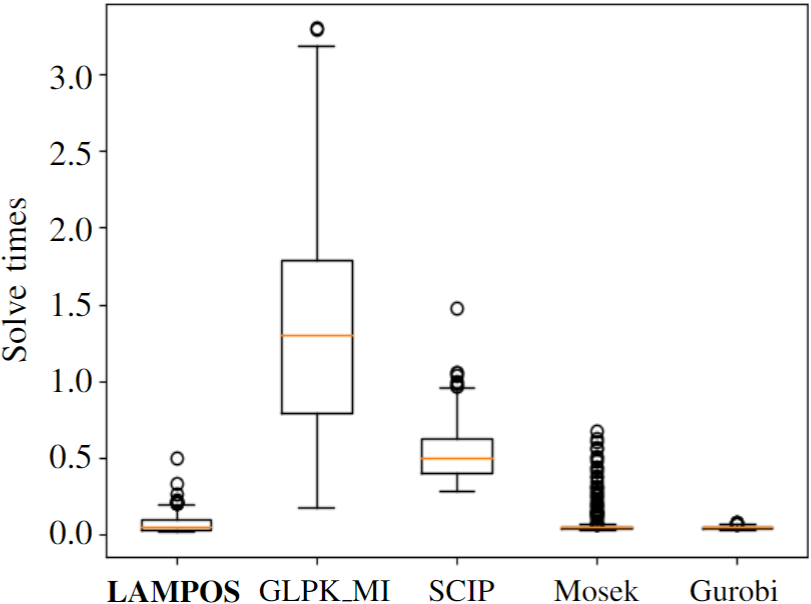}
    \vskip -0.15 true in
    \caption{Comparison with solution times of other solvers  for $N=40$}
    \label{fig:solve_time_cmpN40}
    \vskip -0.15 true in
\end{figure}
\begin{table}[!ht]
\vskip -0.14 true in
    \centering
    \caption{Performance comparison with 50ms solve-time limit}
    \vskip -0.1 true in
    \label{tab:comp}
    \resizebox{0.98\columnwidth}{!}{
    \begin{tabular}[c]{|c|c|c|c|c|c|}
        \hline
            \multirow{2}{*}{Horizon}&\multirow{2}{*}{Metric} & \multicolumn{4}{c|}{Solver}\\\cline{3-6}       
         & &\ LAMPOS & SCIP & Mosek & Gurobi\\
        \hline
        \hline
        \multirow{2}{*}{{\footnotesize{$N$=20}}}& Sub-opt (Avg) & \textbf{0.04} & 0.34 & 0.16 & \textbf{1e-8} \\
        \cline{2-6} &
       Time-out (\%)& $\textbf{0}$ & $\textbf{0}$&  0.2& $0.8$ \\
        \hline
        \multirow{2}{*}{{\footnotesize{$N$=40}}}& Sub-opt (Avg)    & \textbf{0.07} & - & 0.2 & \textbf{1e-10} \\
        \cline{2-6} &
       Time-out (\%)    & 18.6 & 100 & 22.7 &  \textbf{10.8}\\
        \hline
    \end{tabular}}\vskip -0.18 true in
\end{table}
\subsubsection*{\textbf{Discussion}} In Fig. \ref{fig:solve_time_cmp}, \ref{fig:solve_time_cmpN40} for solve-times, we see that LAMPOS outperforms open-source solvers GLPK\_MI, SCIP and is comparable to  Mosek, Gurobi. Table \ref{tab:comp} shows that LAMPOS, Gurobi reliably find high-quality solutions within the time limit compared to SCIP, Mosek. In our experiments, we observed competitive solve-times for LAMPOS when $\tilde{y}(b)=y^\star(b)$, but also quick recovery otherwise by reusing the LP sub-problem information from $\tilde{\mathcal{C}}(b)$ during backup calls. For future investigation, we would like to avoid the parallel solution of the LP sub-problems \eqref{eq:LP_SP}, by exploiting their parametric dependence in $(b, \mathrm{lb}, \mathrm{ub})$. Thus, solutions of the LPs can be predicted in parallel with sub-optimality quantification as in \cite{zhang2020near}, to further decrease solve-times.
\section{Conclusion} 
\vskip -0.08 true in
We proposed a strategy-based prediction framework to solve mp-MILPs online with sub-optimality quantification, and demonstrate it for real-time MIMPC. By exploiting the parametric nature of the optimality certificate for mp-MILPs given by the optimal set of LP sub-problems and an optimal integer solution, we observed favourable performance compared to state-of-the-art MILP solvers. For future work, we aim to include prediction models for solving the parametric LP sub-problems to further improve solve-times.
\vskip -0.1 true in
\bibliographystyle{ieeetr}
\vskip -0.1 true in
\bibliography{root.bib}

\begin{thebibliography}{10}

\bibitem{ioan2021mixed}
D.~Ioan, I.~Prodan, S.~Olaru, F.~Stoican, and S.-I. Niculescu, ``Mixed-integer
  programming in motion planning,'' {\em Annual Reviews in Control}, vol.~51,
  2021.

\bibitem{marcucci2022motion}
T.~Marcucci, M.~Petersen, D.~von Wrangel, and R.~Tedrake, ``Motion planning
  around obstacles with convex optimization,'' {\em arXiv preprint
  arXiv:2205.04422}, 2022.

\bibitem{tokuda2021convex}
S.~Tokuda, M.~Yamakita, H.~Oyama, and R.~Takano, ``Convex approximation for
  ltl-based planning,'' in {\em 2021 IEEE/RSJ International Conference on
  Intelligent Robots and Systems (IROS)}, IEEE, 2021.

\bibitem{heemels2001equivalence}
W.~P. Heemels, B.~De~Schutter, and A.~Bemporad, ``Equivalence of hybrid
  dynamical models,'' {\em Automatica}, vol.~37, 2001.

\bibitem{morari1999model}
M.~Morari and J.~H. Lee, ``Model predictive control: past, present and
  future,'' {\em Computers \& Chemical Engineering}, 1999.

\bibitem{borrelli2017predictive}
F.~Borrelli, A.~Bemporad, and M.~Morari, {\em Predictive control for linear and
  hybrid systems}.
\newblock Cambridge University Press, 2017.

\bibitem{bemporad2000optimal}
A.~Bemporad, F.~Borrelli, and M.~Morari, ``Optimal controllers for hybrid
  systems: Stability and piecewise linear explicit form,'' in {\em Proceedings
  of the 39th IEEE Conference on Decision and Control}, vol.~2, IEEE, 2000.

\bibitem{cimini2017exact}
G.~Cimini and A.~Bemporad, ``Exact complexity certification of active-set
  methods for quadratic programming,'' {\em IEEE Transactions on Automatic
  Control}, vol.~62, 2017.

\bibitem{masti2019learning}
D.~Masti and A.~Bemporad, ``Learning binary warm starts for multiparametric
  mixed-integer quadratic programming,'' in {\em 2019 18th European Control
  Conference (ECC)}, IEEE, 2019.

\bibitem{zhu2020fast}
J.-J. Zhu and G.~Martius, ``Fast non-parametric learning to accelerate
  mixed-integer programming for hybrid model predictive control,'' {\em
  IFAC-PapersOnLine}, vol.~53, 2020.

\bibitem{srinivasan2021fast}
M.~Srinivasan, A.~Chakrabarty, R.~Quirynen, N.~Yoshikawa, T.~Mariyama, and
  S.~Di~Cairano, ``Fast multi-robot motion planning via imitation learning of
  mixed-integer programs,'' {\em IFAC-PapersOnLine}, vol.~54, 2021.

\bibitem{bertsimas2022online}
D.~Bertsimas and B.~Stellato, ``Online mixed-integer optimization in
  milliseconds,'' {\em INFORMS Journal on Computing}, vol.~34, 2022.

\bibitem{cauligi2021coco}
A.~Cauligi, P.~Culbertson, E.~Schmerling, M.~Schwager, B.~Stellato, and
  M.~Pavone, ``Coco: Online mixed-integer control via supervised learning,''
  {\em IEEE Robotics and Automation Letters}, vol.~7, 2021.

\bibitem{zhang2021understanding}
C.~Zhang, S.~Bengio, M.~Hardt, B.~Recht, and O.~Vinyals, ``Understanding deep
  learning (still) requires rethinking generalization,'' {\em Communications of
  the ACM}, vol.~64, 2021.

\bibitem{zhang2020near}
X.~Zhang, M.~Bujarbaruah, and F.~Borrelli, ``Near-optimal rapid mpc using
  neural networks: A primal-dual policy learning framework,'' {\em IEEE
  Transactions on Control Systems Technology}, vol.~29, 2020.

\bibitem{makhorin2008glpk}
A.~Makhorin, ``Glpk (gnu linear programming kit),'' {\em http://www. gnu.
  org/s/glpk/glpk. html}, 2008.

\bibitem{achterberg2009scip}
T.~Achterberg, ``Scip: solving constraint integer programs,'' {\em Mathematical
  Programming Computation}, vol.~1, 2009.

\bibitem{aps2019mosek}
M.~ApS, ``Mosek optimization toolbox for matlab,'' {\em User’s Guide and
  Reference Manual, Version}, vol.~4, 2019.

\bibitem{gurobi}
{Gurobi Optimization, LLC}, ``{Gurobi Optimizer Reference Manual}.''

\bibitem{domahidi2013ecos}
A.~Domahidi, E.~Chu, and S.~Boyd, ``Ecos: An socp solver for embedded
  systems,'' in {\em European control conference (ECC)}, IEEE, 2013.

\end{thebibliography}

\end{document}